\documentclass[runningheads]{llncs}
\usepackage{graphicx}
\usepackage{amsmath,amsfonts,amssymb}
\usepackage[lined,linesnumbered]{algorithm2e}
\usepackage{xcolor}
\usepackage{wrapfig}
\newcommand{\recht}[1]{\operatorname{#1}}

\newcommand{\tOR}{\mathtt{OR}}
\newcommand{\tAND}{\mathtt{AND}}
\newcommand{\tBAS}{\mathtt{BAS}}
\newcommand{\tBCF}{\mathtt{BCF}}

\newcommand{\struc}[1]{\recht{S}_{#1}}
\newcommand{\ch}{\recht{ch}}
\newcommand{\R}[1]{\recht{R}_{#1}}

\newcommand{\PF}{\recht{PF}}
\newcommand{\SCPF}{\recht{SCPF}}
\newcommand{\tzero}{\mathtt{0}}
\newcommand{\tone}{\mathtt{1}}
\newcommand{\BB}{\mathbb{B}}
\newcommand{\prob}{\mathbb{P}}

\newcommand{\vvec}[1]{\left(\begin{smallmatrix}#1\end{smallmatrix}\right)}

\newcommand{\me}{\tau_{\recht{m}}}
\newcommand{\ee}{\tau_{\recht{e}}}

\newcommand{\hlbox}[1]{%
  \smallskip\begin{center}
  \fboxrule1pt\fboxsep3pt\fcolorbox{black!45}{black!8}{%
  \begin{minipage}{.96\linewidth}#1\end{minipage}}
  \end{center}\smallskip}

\allowdisplaybreaks

\begin{document}
\title{Quantitative analysis of attack-fault trees via Markov decision processes\thanks{This research has been partially funded by ERC Consolidator grant 864075 CAESAR and the European Union’s Horizon 2020 research and innovation programme under the Marie Skłodowska-Curie grant agreement No. 101008233.}}
%
%
\author{Milan Lopuhaä-Zwakenberg}
%
\authorrunning{Milan Lopuhaä-Zwakenberg}
%
\institute{University of Twente\\
\email{m.a.lopuhaa@utwente.nl}}
\maketitle              
\begin{abstract} Adequate risk assessment of safety critical systems needs to take both safety and security into account, as well as their interaction. A prominent methodology for modeling safety and security are attack-fault trees (AFTs), which combine the well-established fault tree and attack tree methodologies for safety and security, respectively. AFTs can be used for quantitative analysis as well, capturing the interplay between safety and security metrics. However, existing approaches are based on modeling the AFT as a priced-timed automaton. This allows for a wide range of analyses, but Pareto analsis is still lacking, and analyses that exist are computationally expensive. In this paper, we combine safety and security analysis techniques to introduce a novel method to find the Pareto front between the metrics \emph{reliability} (safety) and \emph{attack cost} (security) using Markov decision processes. This gives us the full interplay between safety and security while being considerably more lightweight and faster than the automaton approach. We validate our approach on a case study of cyberattacks on an oil pipe line.

\keywords{Attack-fault trees  \and Quantitative analysis \and Markov decision processes \and Reliability \and Security}
\end{abstract}
\section{Introduction}

\noindent \textbf{Safety-security co-analysis.} As high-tech systems become increasingly advanced, so too does their potential for vulnerabilities, either from accidental failure (safety) or malicious attackers (security). In interconnected systems, safety and security become more and more intertwined: for instance, a ransomware attack on a hospital can lead to the unavailability of essential systems, resulting in the loss of patient lives \cite{argaw2019state}. A comprehensive risk analysis that includes both safety and security is essential to guarantee a system's continued functioning.

\noindent \textbf{Attack-fault trees.} A common model for safety-security analysis is the \emph{attack-fault tree} (AFT) \cite{kumar2017quantitative}, which combines the well-established  models of fault trees (FTs) for safety \cite{ruijters2015fault} with attack trees (ATs) for security \cite{schneier1999attack}. An AFT is a directed acyclic graph whose root represents system failure. Its leaves denote basic events, which are either \emph{basic attack steps} (BASs) or \emph{basic component failures} (BCFs), depending on whether they represent attacker actions or accidental failures. The intermediate nodes are AND/OR-gates, that are compromised when all (resp. one) of their children are. An example is presented in Fig.~\ref{fig:bank}.

\begin{wrapfigure}[21]{r}{5cm}
\vspace{-2.2em}
    \centering
\includegraphics[width=5cm]{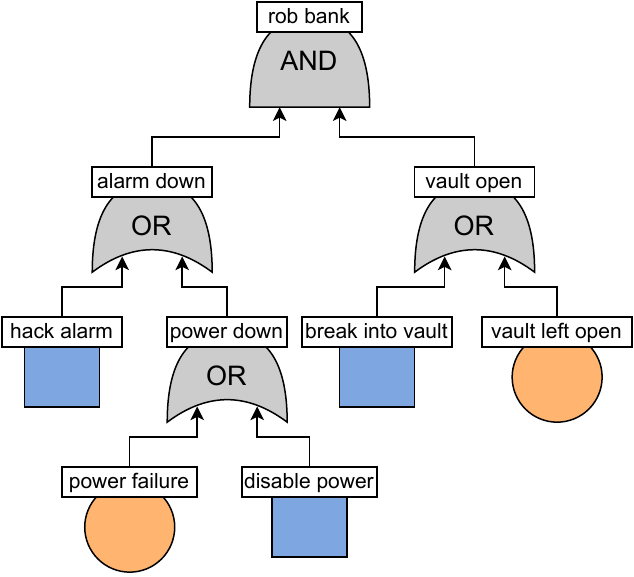}
\vspace{-2.2em}
\caption{An example of an AFT. In order to rob a bank, an attacker needs to both a disabled alarm and an open vault. The former can be accomplished by either hacking the alarm or by a power outage, which can be both accidental, or the result of an attack. The open vault can be a result of cracking the vault, or the vault can be left open by accident.}
 \label{fig:bank}
\end{wrapfigure}

\noindent \textbf{Quantitative analysis.} AFTs can be used for qualitative analysis, cataloguing which combinations of attacks and failure lead to system compromise. However, a more detailed vulnerability assessment is obtained by quantitative analysis, in which each BAS/BCF is assigned parameter values, from which risk metrics of the system can be deduced. For instance, in \cite{fovino2009integrating}, each BAS/BCF is assigned a probability value, which together determine the system's failure probability. However, safety and security often have different metrics, and conflating these into a single metric does not take the differences between safety and security, and their modeling formalisms \cite{budde2021attack}, into account.

AFTs can also be modeled as timed automata \cite{kumar2017quantitative,andre2019parametric}; these can then be given a wide number of parameters, such as attacker cost, failure rate, detection probability, etc., upon which statistical model checking tools can calculate a range of metrics. This method is very flexible and allows for thorough analysis, but it has two downsides: First, the automata models are large and intricate, so state space becomes unfeasibly large for larger AFTs \cite{volk2017fast}. Second, in these works metrics are not directly defined, but arise as a consequence of the automaton model. This makes it hard to judge to what extent the computed metric actually corresponds to the system's vulnerability. For instance, the automaton model of \cite{kumar2015quantitative} can be argued to not compute the actual minimal attack time of an attack tree, but just a lower bound \cite{lopuhaa2023attack}. Thus there is a need of quantitative AFT analysis methods that are lightweight, with clear, mathematically defined metrics, that take the dual safety-security nature into account.

\noindent \textbf{Cost-probability Pareto front.} In this paper we consider the setting where every BCF is assigned a failure probability and each BAS is assigned a cost to the attacker, which can represent monetary cost, resources spent, time, etc. The attacker chooses which BASs to activate such that the top level event is compromised with a high probability, while keeping their total costs low. This situation is a straightforward combination of FT unreliability analysis \cite{ruijters2015fault} and AT metric analysis \cite{lopuhaa2022efficient}. However, the situation is considerably more complicated than the separate safety and security analyses. First, the attacker's two goals are often conflicting: in order to ensure a higher compromise probability, they will need to spend additional costs. Therefore, there is not one optimal attack, but rather an entire Pareto front, and the goal of the analysis is to find all Pareto-optimal attacker strategies.  Existing AFT analyses do not compute the Pareto front, instead analyzing individual strategies for multiple metrics 
\cite{kumar2017quantitative,andre2019parametric}.

Second, it is important to consider to what extent the attacker can observe the failure of BCFs before deciding on an action. All BASs may occur first, or all BCFs, but more complicated relations are possible. In the last two cases, an attacker's strategy is not simply a set of activated BASs as in standard AT analysis; instead, it gives an action for each possible observed set of failed BCFs.

\noindent \textbf{Contributions.} This paper's main contribution is to give a fast algorithm for Pareto analysis on attack-fault trees. This goal is achieved as follows.
First, we give a formal definition of cost-probability Pareto front problems that includes the timing relation between BASs and BCFs. These can be generalized beyond AFTs to arbitrary Boolean functions. Since attacker strategies depend on the failure of observed BCFs, which are random events, the total cost of a strategy is a random variable itself. We consider two cases: the attacker either aims to minimize the maximal cost (for when an attacker only attempts one attack, with a set maximum budget), or to minimize the expected cost (for when an attacker attempts attacks on many similar systems, and wants to do so cost-effectively).

The AFT is analyzed by transforming it into a binary decision diagram (BDD). This technique is used in FT and AT analysis \cite{ruijters2015fault,lopuhaa2022efficient}, but the BDD of an AFT has a richer structure: it can be interpreted as a Markov Decision Process (MDP), which are widely studied in statistical model checking. However, our setting differs from standard MDP analysis: the resulting MDP is acyclic, we consider two metrics simultaneously, and maximum cost cannot be expressed as an expected value. Approaches to handle these differences exist \cite{bellman1957markovian,etessami2008multi,de2005model}, but they cannot be combined directly. In this paper, we give a new, straightforward solution to this problem, combining the ideas of the existing literature, and prove its correctness. This results in a fast algorithm, of which we show its validity by applying it to the safety-security analysis of an oil pipeline \cite{kumar2017quantitative}. All proofs are presented in the appendices.

\subsection{Related work}

To address safety and security, ATs and FTs have been combined in different ways \cite{nicoletti2023model}. In \cite{fovino2009integrating}, ATs are used to determine the outcome of some BCFs in FTs. Both BASs and BCFs are then assigned probabilities, which is used to calculate the reliability of the system. This methodology is straightforward to implement, but cannot model more complicated safety-security interdependencies. \cite{steiner2013combination} uses the same modeling, but assigns to the minimal cut sets both a probability (safety) and a security rating (low/medium/high), which allows for a basic Pareto analysis, although rigorous security metric analysis is missing.

Dynamic AFTs are described in \cite{kumar2015quantitative,andre2019parametric}: these do also include the sequential $\tAND$ and priority $\tAND$ gates from dynamic ATs/FTs. These AFTs are translated to timed automata, allowing for a rich analysis of many metrics. The downside of this approach is that the gates' behaviour is only informally described, making it difficult to determine to what extent the metric computed by the automata model corresponds with what one would expect \cite{lopuhaa2023attack}. Such a translation is also computationally expensive. Finally, no existing works performs a full Pareto analysis; for a detailed comparison between approaches see Section \ref{ssec:compare}.

Since there are multiple relevant security metrics, there are multiple works on finding Pareto fronts in ATs \cite{fila2019efficient,lopuhaa2022efficient,lopuhaa2023cost}. For FTs this is less relevant as one typically only considers reliability, or derived metrics such as \emph{mean time to fail}. In Markov decision processes multi-objective optimization can be achieved by multi-objective linear programming \cite{chatterjee2006markov,etessami2008multi}.

\section{Attack-fault tree semantics}

In this section, we define attack-fault trees as used in this paper. We take what is arguably the simplest model, where the only intermediate gates are $\tAND$- and $\tOR$-gates. Many extensions exist, including time-dependent gates \cite{kumar2017quantitative}, but our gate types are enough to define the interplay between costs and (static) probability.

\begin{definition}
An \emph{(augmented) attack-fault tree} (AFT) is a tuple $T = (V,E,\gamma)$ where $(V,E)$ is a rooted directed acyclic graph, and $\gamma\colon V \rightarrow \{\tOR,\tAND,\tBAS,\tBCF\}$ satisfies $\gamma(v) \in \{\tBCF,\tBAS\}$ if and only if $v$ is a leaf.
\end{definition}

For an AFT $T$ we furthermore define its set of BCFs $F_T = \{v \in V \mid \gamma(v) = \tBCF\}$ and its set of BASs $A_T = \{v \in V \mid \gamma(v) = \tBAS\}$. The root of $T$ is denoted $\R{T}$, and the children of a node $v$ form the set $\ch(v) = \{w \in V \mid (v,w) \in E\}$.

Denote the Boolean domain $\{\tzero,\tone\}$ by $\BB$. A \emph{safety-security event} is the simultaneous occurrence of some BCFs and BASs. We can model this as a pair $(x,y)$, where $x \in \mathcal{X}_T := \BB^{F_T}$ is a binary vector denoting which BCFs occur, and where $y \in \mathcal{Y}_T := \BB^{A_T}$ denotes which BASs occur. To what extent this leads to a compromised system depends on the AFT: an $\tOR$-gate is compromised if any of its children is compromised, while an $\tAND$-gate is compromised if the system as a whole is compromised. This is captured by the following definition.

\begin{definition}
Let $T = (V,E,\gamma)$ be an AFT. The \emph{structure function} of $T$ is the function $\struc{T}\colon V \times \mathcal{X}_T \times \mathcal{Y}_T \rightarrow \BB$ given by
\[
\struc{T}(v,x,y) = \begin{cases}
x_v, & \textrm{ if $\gamma(v) = \tBCF$,}\\
y_v, & \textrm{ if $\gamma(v) = \tBAS$,}\\
\bigvee_{w \in \ch(v)} \struc{T}(w,x,y) & \textrm{ if $\gamma(v) = \tOR$,}\\
\bigwedge_{w \in \ch(v)} \struc{T}(w,x,y) & \textrm{ if $\gamma(v) = \tAND$.}\\
\end{cases}
\]
$(x,y)$ \emph{compromises $v$} if $\struc{T}(v,x,y) = \tone$. It \emph{compromises $T$} if it compromises $\R{T}$.
\end{definition}

\section{Example of AFT analysis} \label{sec:example}

Before we fully dive into the mathematical framework of AFT analysis, we give a small example to showcase the subtleties involved. Consider the AFT of Fig.~\ref{fig:exa}. It consists of two components, each of which is only disabled if both a failure occurs (probability 0.5) and an attack succeeds (attacker cost 10). Disabling one component is enough to shut down the system; the attacker's goal is to disable the system with a probability that is as high as possible, spending as few resources as possible. As these are conflicting goals, the attacker's optimal choices are given by a Pareto front between cost and probability.

First, assume that the attacker needs to choose their attack before observing any failures. In this case, the attacker has three strategies. All of these are Pareto optimal, as each one has both a higher cost and a higher compromise probability than the previous one:

\begin{itemize}
\item Do not attack at all, for a cost of $0$ and a compromise probability of $0$.
\item Attack one of the components. This has a cost of $10$, and a resulting compromise probability of $0.5$.
\item Attack both components for a cost of $20$. Now the system is not compromised only if both components do not fail, so the compromise probability is $0.75$.
\end{itemize}

\begin{wrapfigure}[11]{r}{5cm}
    \vspace{-3.5em}
        \centering
    \includegraphics[width=4.5cm]{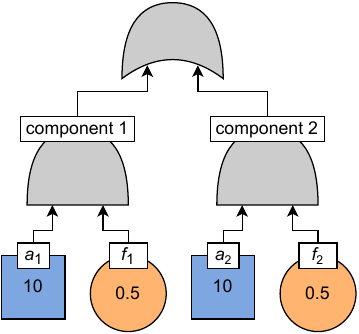}
    \vspace{-1.5em}
    \caption{The example AFT discussed in Section \ref{sec:example}.} \label{fig:exa}
    \end{wrapfigure}

Now consider the situation where the attacker gets to observe whether any failures occur. In this case, the attacker has many different strategies: there are $4$ possible failures $x \in \mathcal{X}_T$ that the attacker can observe, and to choose a strategy is to choose an attack $y \in \mathcal{Y}_T$ for each observed failure vector. As there are $4$ possible attacks, the number of strategies is $4^4 = 256$. Thus even in this small example, a full analysis of all strategies becomes quite cumbersome.

A further complication is that we must be precise about what we mean by the cost of a strategy, as the cost the attacker incurs depends on $x$, which is random. We can define cost in two ways:
First, we can take the maximum cost over all possible $x$. This is a useful metric in which the attacker performs a single attack on a system, and is constrained by a budget. The alternative is to consider the expected cost. This is useful when the attacker performs many different attacks on the same or similar systems, and wants to do so in a cost-efficient way. 

In general, these approaches lead to different optimal strategies. To find these, one calculates the success probability and maximum/expected cost of each and then finds the Pareto front. It turns out that the following strategies are optimal for both expected and maximum cost:
\begin{description}
\item[$\sigma^1$.] Do no attack regardless of $x$, for a maximum/expected cost of $0$ and a probability of $0$.
\item[$\sigma^2$.] perform attacks depending on $x$: if $x_1 = x_2 = \tzero$, do nothing. Otherwise, choose an $x_i$ equal to $\tone$ and perform attack $a_i$. This attack has a success probability of $0.75$, as it succeeds if there is a nonzero $f_i$. Its maximum cost is $10$, as we activate at most one $a_i$, and its expected cost is $7.5$.
\end{description}

For maximum cost, the two given attacks form our Pareto front. For expected cost, one can also consider mixed strategies, picking $\sigma^1$ with probability $p$ and picking $\sigma^2$ with probability $1-p$. This mixed strategy has success probability and cost both $0.75(1-p)$; each choice of $p$ is Pareto-optimal. For maximum cost mixed strategies do not help: for $0 < p < 1$ the mixed strategy above has maximum cost $10$, and probability $<0.75$, which is less efficient than pure $\sigma^2$.

The two scenarios discussed here (either all failures first, or all attacks first) are not the only possiblities. In general, in order to define a strategy, we need to know for each attack $a$ the set of $Q_a$ of basic events that occur before the attacker has to decide on $a$. In this section, we have considered $Q_{a_1} = Q_{a_2} = \varnothing$ and $Q_{a_1} = Q_{a_2} = \{f_1,f_2\}$, but these are not the only possibilities.

\section{Quantitative analysis}

In this section, we essentially define quantitative analysis of AFTs. However, a key insight is that these definitions do not depend on the AFT's graph structure, but only on the structure function. Therefore, instead of defining quantitative analysis for AFTs, we define it for general Boolean functions. 

\subsection{Strategies}

In this subsection we generalize the AFT setting of Section \ref{sec:example}, and their strategies.

\begin{definition}
A \emph{scenario} is a tuple $\Lambda = (F,A,\varphi,Q)$, where $F$ and $A$ are finite sets; $\varphi\colon \mathcal{X} \times \mathcal{Y} \rightarrow \BB$ is a Boolean function, where $\mathcal{X} = \BB^F$ and $\mathcal{Y} = \BB^A$; and $Q = (Q_a)_{a \in A}$ is a collection of subsets of $F$, such that the set $\{Q_a \mid a \in A\}$ is linearly ordered by $\subseteq$.
\end{definition}

In the context of AFTs, one takes $A = A_T$, $F  = F_T$, and $\varphi = \struc{T}(\R{T},\_,\_)$. For a BAS $a \in A_T$, the set $Q_a$ is the set of BCFs whose failure status is known to the attacker at the time of deciding whether to perform $a$. The last property allows for the possibility that $Q_a = Q_{a'}$ for $a,a' \in A$. The $Q_a$ are required to be linearly ordered by $\subseteq$ for temporal consistency: when the attacker can decide about $a$ later than they have to decide about $a'$, we have $Q_{a'} \subseteq Q_a$.

A strategy $\sigma$ denotes an attacker's decision whether or not to perform $a$, having observed the status of $Q_a$. This is captured by a map $\sigma_a\colon\BB^{Q_a} \rightarrow \BB$, where $\sigma_a(x) = \tone$ means that the attacker performs $\sigma$ upon observing $x$. This is captured by the following definition:

\begin{definition}
Let $\Lambda = (F,A,\varphi,Q)$ be a scenario. A \emph{strategy} for $\Lambda$ is a collection $\sigma = (\sigma_a)_{a \in A}$ where each $\sigma_a$ is a map $\sigma_a\colon \BB^{Q_a} \rightarrow \BB$. A strategy $\sigma$ induces a map $\hat{\sigma}\colon \mathcal{X} \rightarrow \mathcal{Y}$ given by $\hat{\sigma}(x)_a = \sigma_a((x_f)_{f \in Q_a}).$ The set of strategies for $\Lambda$ is denoted $\Sigma_{\Lambda}$.
\end{definition}

\begin{example} \label{ex:AFT1}
Consider the example of Section \ref{sec:example}, assuming the attacker observes failures before acting. Then $F = \{f_1,f_2\}$, $A = \{a_1,a_2\}$, $\varphi = (f_1 \wedge a_1) \vee (f_2 \wedge a_2)$ and $Q_{a_1} = Q_{a_2} = \{f_1,f_2\}$. Write $x \in \mathcal{X}$ as $x_{f_1}x_{f_2}$ and $y \in \mathcal{Y}$ as $y_{a_1}y_{a_2}$; then for the strategies $\sigma^1$ and $\sigma^2$ defined there, $\hat{\sigma}^1$ and $\hat{\sigma}^2$ are given by
\begin{align*}
\forall x\colon\hat{\sigma}^1(x) &= \tzero\tzero & \textrm{and} && \hat{\sigma}^2(\tzero\tzero) &= \tzero\tzero, & \hat{\sigma}^2(\tzero\tone) &= \tzero\tone, &\hat{\sigma}^2(\tone\tzero) = \hat{\sigma}^2(\tone\tone) &= \tone\tzero.
\end{align*}
\end{example}

Note that we have only defined strategies, and not how sensible a strategy is to an attacker. This comes in the next section, where we define metrics that define how good a strategy is.

\subsection{Strategy metrics}

In order to measure the effectiveness of strategies, we need to assign metrics to the variables $f \in F$, $a \in A$ of a scenario $(F,A,\varphi,Q)$. To each $f \in F$ we assign a failure probability $\pi_f \in [0,1]$, i.e., the probability that $f$ is compromised. Meanwhile, to each $a \in A$ we assign a cost $\gamma_a \in [0,\infty]$, denoting the cost the attacker needs to spend in order to compromise $a$.

\begin{definition}
A \emph{quantified scenario} is a tuple $\hat{\Lambda} = (\Lambda,\pi,\gamma)$, where $\Lambda = (F,A,\varphi,Q)$ is a scenario, the \emph{probability vector} $\pi$ is an element of $[0,1]^F$ and the \emph{cost vector} $\gamma$ is an element of $[0,\infty]^A$.
\end{definition}

There are many attack tree metrics considered beyond cost in the literature \cite{lopuhaa2022efficient}. We restrict our focus to cost for two reasons: first, many other metrics such as time or probability can be rephrased as costs. Second, taking a semiring-based approach that generalizes many metrics at once is difficult in our scenario, since the notion of expected value is not defined for general semirings; thus we would not be able to define the analog of expected cost.

In our model, attacks do not have a success probability: $a$ always succeeds when the attacker pays cost $\gamma_a$. This is not a restriction: if we want to model a system in which $a$ only succeeds with probability $\pi_a$, we can replace $a$ with $\tAND(a',f')$ with $\gamma_{a'} = \gamma_a$, $\pi_{f'} = \pi_a$, and $f' \notin Q_{a'}$. Thus our choice to only assign probabilities to BCFs does not limit our generality.

The probability and costs vectors allow us to assess a strategy's merit. A strategy's probability $\hat{\pi}(\sigma)$ is the probability that it leads to $\varphi = \tone$, under the assumption that each $f \in F$ is compromised with probability $\pi_f$; we assume these are independent events, as is standard in fault tree analysis \cite{ruijters2015fault}. Meanwhile, the total cost of a strategy is a random variable, depending on which $f \in F$ become compromised. Thus we consider both the maximum cost $\hat{\gamma}_{\recht{m}}(\sigma)$ that a strategy can incur over all possible failures in $F$, and the expected cost $\hat{\gamma}_{\recht{e}}(\sigma)$. In the definition below, $\mathsf{v}\sim \recht{Ber}(p)$ denotes a Bernoulli variable, i.e. $\mathsf{v} \in \BB$ and $\prob(\mathsf{v} = \tone) = p$ and $\prob(\mathsf{v} = \tzero) = 1-p$. 

\begin{definition}
Let $\hat{\Lambda}$ be a quantified scenario. For $\sigma \in \Sigma_{\Lambda}$, the \emph{probability} $\hat{\pi}(\sigma)$, \emph{maximal cost} $\hat{\gamma}_{\recht{m}}(\sigma)$ and \emph{expected cost} $\hat{\gamma}_{\recht{e}}(\sigma)$ are defined as follows. Let $\mathsf{x}$ be a random variable on $X$ with independent coefficients $\mathsf{x}_f \sim \recht{Ber}(\pi_f)$. Then
\begin{align*}
\hat{\pi}(\sigma) &= \prob_{\mathsf{x}}(\varphi(\mathsf{x},\hat{\sigma}(\mathsf{x})) = \tone), &
\hat{\gamma}_{\recht{m}}(\sigma) &= \max_{x \in X} \sum_{\substack{a \in A\colon\\ \hat{\sigma}(x)_a = \tone}} \gamma_a, &
\hat{\gamma}_{\recht{e}}(\sigma) &= \mathbb{E}_{\mathsf{x}}\sum_{\substack{a \in A\colon\\ \hat{\sigma}(\mathsf{x})_a = \tone}} \gamma_a.
\end{align*}
\end{definition}
\begin{example} \label{ex:AFT2}
We continue Example \ref{ex:AFT1}. We extend the $\Lambda$ defined there with the probabilities and costs of Section \ref{sec:example}: $\pi_{f_1} = \pi_{f_2} = 0.5$ and $\gamma_{a_1} = \gamma_{f_2} = 10$. Then $\mathsf{x}$ has four possibilities $\tzero\tzero,\tzero\tone,\tone\tzero,\tone\tone$. We calculate their probabilities: we have $\mathbb{P}(\mathsf{x} = \tone\tone) = \mathbb{P}(\mathsf{x}_{f_1} = \tone)\cdot \mathbb{P}(\mathsf{x}_{f_2} = \tone) = \pi_{f_1}\pi_{f_2} = 0.25$. In fact, all four outcomes have equal probability. Since $\varphi(x,\tzero\tzero) = \tzero$ for all $x$ and since $\sigma^1$ never activates any attacks, we have $\hat{\pi}(\sigma^1) = 0$. For the same reason we have $\hat{\gamma}_{\recht{m}}(\sigma^1) = \hat{\gamma}_{\recht{e}}(\sigma^1) = 0$. Meanwhile, from the definition of $\varphi$ and $\hat{\sigma}^2$ from Example \ref{ex:AFT1} we see that $\varphi(x,\hat{\sigma}^2(x)) = \tzero$ only if $x = \tzero\tzero$, which happens with probability 0.25; hence $\hat{\pi}(\sigma^2) = 0.75$. Furthermore, the cost of a strategy, upon observing $x$, is
\begin{align*}
\gamma_{a_1}\hat{\sigma}^2(x)_{a_1}+\gamma_{a_2}\hat{\sigma}^2(x)_{a_2} &= \begin{cases}
0, & \textrm{ if $x = \tzero\tzero$,}\\
10, & \textrm{ otherwise.}
\end{cases}
\end{align*}
Hence $\hat{\gamma}_{\recht{m}}(\sigma^2) = 10$ and $\hat{\gamma}_{\recht{e}}(\sigma^2) = 7.5$.
\end{example}

The attacker's aim is to choose $\sigma$ such that $\hat{\pi}(\sigma)$ is as large as possible, while keeping $\hat{\gamma}_{\recht{m/e}}(\sigma)$ as small as possible. These are conflicting concerns, and the solution to this problem is to find the Pareto front of optimal strategies. This is expressed as follows. Let $\mathcal{D} = [0,1] \times [0,\infty]$. We can define the following:

\begin{definition}
Let $\hat{\Lambda}$ be a quantified scenario. Define  $\me,\ee\colon \Sigma_{\Lambda} \rightarrow \mathcal{D}$ by
\begin{align*}
\me(\sigma) &= \vvec{\hat{\pi}(\sigma)\\\hat{\gamma}_{\recht{m}}(\sigma)},& \ee(\sigma) &= \vvec{\hat{\pi}(\sigma)\\\hat{\gamma}_{\recht{e}}(\sigma)}.
\end{align*}
\end{definition}

Let $\sqsubseteq$ be the partial order on $\mathcal{D}$ given by $\vvec{c\\p}\sqsubseteq \binom{c'}{p'}$ iff $c \leq c'$ and $p \geq p'$. In the maximal cost setting, the attacker's aim is to find strategies $\sigma$ for which $\me(\sigma)$ is minimal w.r.t. $\sqsubseteq$. This can be formalized by a Pareto front:

\begin{definition}
Let $I \subseteq \mathcal{D}$. The \emph{Pareto Front} of $I$ (with respect to $\sqsubseteq$) is defined as $\PF(I) = \{d \in I \mid \forall d' \in I: d' \not \sqsubset d\}$. For a quantified scenario $\hat{\Lambda}$, its \emph{probability vs. maximal cost Pareto front} $\recht{PMC}(\hat{\Lambda})$ is the Pareto front $\PF(\me(\Sigma_{\Lambda}))$.
\end{definition}

In the expected cost setting, the Pareto front is more involved. The set $\Sigma_{\Lambda}$ only consists of \emph{pure strategies}, i.e., strategies in which the attacker's action is deterministic given $\mathsf{x}$. The attacker can also perform mixed strategies, assigning a probability $\lambda_{\sigma}$ to each strategy $\sigma$. The probability that this mixed strategy  results in $\varphi = \tone$ is $\sum_{\sigma} \lambda_{\sigma} \hat{\pi}(\sigma)$, while the expected cost is $\sum_{\sigma} \lambda_{\sigma} \hat{\gamma}_{\recht{e}}(\sigma)$. In other words, when also considering mixed strategies, we do not just get $\me(\Sigma_{\Lambda})$, but its entire convex hull $\recht{Conv}(\me(\Sigma_{\Lambda}))$. This is a two-dimensional polyhedron, and its Pareto front consists of a finite number of vertices that represent pure strategies, and the line segments between them. Thus we only need to find the Pareto-optimal vertices; Formally, this can be defined as follows.

\begin{definition}
Let $I \subset \mathcal{D}$. For $d_1,d_2 \in I$, let $\ell(d_1,d_2)$ be the line segment connecting them. The \emph{strictly convex Pareto Front} of $I$ is defined as $\SCPF(I) = \{d \in I \mid \forall d_1,d_2 \in I, \forall d' \in \ell(d_1,d_2): d' \not \sqsubset d\}$. For a quantified scenario $\hat{\Lambda}$, its \emph{probability vs expected cost Pareto front} $\recht{PEC}(\hat{\Lambda})$ is the strictly convex Pareto front $\SCPF(\ee(\Sigma_{\Lambda}))$.
\end{definition}
The main goal of this paper is then to find fast solutions to the following:

\hlbox{{\bf Problem.}
Given a quantified scenario $\hat{\Lambda}$, calculate $\recht{PMC}(\hat{\Lambda})$ and $\recht{PEC}(\hat{\Lambda})$.}
This problem is NP-hard, as calculating FT failure probability is NP-hard \cite{lopuhaa2023fault}; hence it is important to find computationally efficient solutions.

\begin{example}
Suppose $\varphi$ is the structure function of an attack tree. Hence $F = \varnothing$ and $\varphi\colon \BB^A \rightarrow \BB$ is nondecreasing and nonconstant. In this case, each $\sigma_a$ is just an element of $\BB$, and we can consider $\sigma$ to be an element of $\BB^A = Y$. We get
\begin{align*}
\hat{\pi}(\sigma) &= \varphi(\sigma) \in \{0,1\}, &
\hat{\gamma}_{\recht{m}}(\sigma) = \hat{\gamma}_{\recht{e}}(\sigma) &= \sum_{a \in A} \sigma_a \gamma_a.
\end{align*}
In $\recht{PMC}(\hat{\Lambda})$ only two first coordinates are possible, so $|\recht{PEC}(\hat{\Lambda})| = 2$; for each value of $\hat{\pi}(\sigma)$ the Pareto front has one element, corresponding to the $\sigma$ that minimizes $\hat{\gamma}_{\recht{m}}(\sigma)$ given $\varphi(\sigma)$. For $\varphi(\sigma) = \tzero$ this is $\sigma = \vec{\tzero}$. For $\varphi(\sigma) = \tone$, the corresponding second coordinate is the minimal cost of a succesful attack as defined in \cite{lopuhaa2022efficient}. Furthermore $\recht{PEC}(\hat{\Lambda}) = \recht{PMC}(\hat{\Lambda})$.
\end{example}

\begin{example}
Suppose $\varphi$ is the structure function of a fault tree. Hence $A = \varnothing$ and $\varphi\colon \BB^F \rightarrow \BB$ is nondecreasing and nonconstant. In this case, $\Sigma_{\Lambda} = \{\varnothing\}$ as there is no choice of strategy. For this `strategy' one has
\begin{align*}
\hat{\pi}(\varnothing) &= \prob(\varphi(\vec{F}) = \tone), &
\hat{\gamma}_{\recht{m}}(\varnothing) = \hat{\gamma}_{\recht{e}}(\varnothing) &= 0.
\end{align*}
Thus $\recht{PMC}(\hat{\Lambda})$ and $\recht{PEC}(\hat{\Lambda})$ consist of a single element of $\mathcal{D}$, whose first coordinate is the \emph{unreliability} of the fault tree \cite{ruijters2015fault}.
\end{example}

A difference between our general setting and the AFT setting is that the structure function of an AFT is required to be increasing and nonconstant, while our function $\varphi$ does not have such constraints. In principle, this makes it possible to model more complex safety-security interactions, in which, for example, failures stop attacks from propagating. Such behaviour is called \emph{antagonism} \cite{kriaa2015survey}, which can be incorporated into AFTs e.g. by adding a \texttt{NOT}-gate \cite{nicoletti2023model}.

\subsection{Relation to quantitative AFT analysis methods} \label{ssec:compare}

Before diving into our method of computing $\recht{PMC}(\hat{\Lambda})$ and $\recht{PEC}(\hat{\Lambda})$, we first discuss to what extent our model and metrics differ from earlier approaches to quantitative AFT analysis: the two relevant works are \cite{kumar2017quantitative} and \cite{andre2019parametric}. Both do not give a formal definition of their metrics; instead, they are obtained by translating the AFT to an automata model and analyzing these using model checking tools. 

\cite{kumar2017quantitative} allows BASs/BCFs to be endowed with a wide range of parameters, including cost (fixed and variable), damage, and success probabilities, both as constant probabilities and as failure rates. Indeed, one of the major differences to our model is that their AFTs are dynamic, and are able to describe the behaviour of the system over time. Combined with the expressivity of the UPPAAL query language \cite{bengtsson1996uppaal}, this allows for a detailed analysis. However, it does not come with the means to find (Pareto-)optimal strategies. Instead, individual strategies are analyzed for multiple metrics.

\cite{andre2019parametric} takes a different approach, assigning cost and time parameters to each BAS/BCF and then using the IMITATOR model checker to determine whether a strategy with given cost/time values exists. In principle, this gives a (potentially time-consuming) method to find the cost/time Pareto front. However, their approach does not incorporate probability, which is precisely the factor that makes our definition of strategy so involved.

We conclude that existing methods are not able to compute $\recht{PMC}(\hat{\Lambda})$ and $\recht{PEC}(\hat{\Lambda})$, as they either are concerned with different metrics on AFTs, or are not able to synthesize strategies.

\section{Boolean Decision Diagrams}

In this section, we briefly review the theory of binary decision diagrams, which form compact representations of Boolean functions \cite{bryant1986graph}. We will apply this to the Boolean functions of quantified scenarios to calculate $\recht{PMC}(\hat{\Lambda})$ and $\recht{PEC}(\hat{\Lambda})$.

\begin{definition}
Let $C$ be a finite set, and let $<$ be a strict linear order on $C$. A \emph{reduced ordered binary decision diagram} (ROBDD) on $(C,<)$ is a tuple $B = (V,E,t_V,t_E)$ where:
\begin{enumerate}
\item $(V,E)$ is a rooted directed acyclic graph where every nonleaf has exactly two children.
\item $t_V\colon V \rightarrow C \cup \BB$ satisfies $t_V(v) \in \BB$ iff $v$ is a leaf.
\item $t_E\colon E \rightarrow \{\tzero,\tone\}$ is such that the two outgoing edges of each nonleaf $v$ have different values; the two children are denoted $c_{\tzero}(v)$ and $c_{\tone}(v)$ depending on the edge values.
\item If $(v,v')$ is an edge and $t_V(v) = F_i$, $t_V(v') = F_{i'}$, then $i < i'$.
\item If $v,v'$ are two nodes such that $t_V(v) = t_V(v')$, and for nonleaves furthermore $c_{\tzero}(v) = c_{\tzero}(v')$ and $c_{\tone}(v) = c_{\tone}(v')$, then $v = v'$.
\end{enumerate}
\end{definition}

An ROBDD $B$ represents a Boolean function $\varphi\colon\BB^C \rightarrow \BB$ as follows: let $z \in \BB^C$, and start from the root $\R{B}$. At a non-leaf node $v \in V$, we proceed to $c_{\tzero}(v)$ or $c_{\tone}(v)$, depending on the value of $z_{t_V(v)}$ (recall that $t_V(v) \in C$). When we arrive at a leaf $v$, we define $\varphi(z) = t_V(v) \in \BB$. An important fact is that every $\varphi$ can be represented this way \cite{bryant1986graph}. Different choices of $<$ result in different ROBDDs representing the same $\varphi$, and in general, the size of the ROBDD heavily depends on $<$. Finding the optimal $<$ is NP-hard, but heuristics exist \cite{bouissou1997bdd}.

\section{Analysis of quantified scenarios}

In this section, we show how ROBDDs can be used to calculate $\recht{PMC}(\hat{\Lambda})$ and $\recht{PEC}(\hat{\Lambda})$. First, we describe how to create the ROBDD of a quantified scenario. The first step is to transform the time relation between failures and attacks implicit in $Q$ into a partial order. This can be done as follows:

\begin{definition}
Let $\Lambda = (F,A,\varphi,Q)$ be a scenario. Then we define the relation $\prec_{\Lambda}$ on $F \cup A$ as follows, for $f \neq f' \in F$ and $a \neq a' \in A$:
\begin{align*}
a \prec_{\Lambda} a' & \textrm{ iff $Q_a \subset Q_{a'}$,} & a \prec_{\Lambda} f & \textrm{ iff $f \notin Q_a$,}\\
f \prec_{\Lambda} a & \textrm{ iff $f \in Q_a$,} & f \prec_{\Lambda} f' & \textrm{ iff $\exists a \in A. f \in Q_a \not \ni f'$.}
\end{align*}
\end{definition}

\begin{lemma} \label{lem:spo}
The relation $\prec_{\Lambda}$ is a strict partial order.
\end{lemma}

We use this order to define the ROBDD of $\Lambda$.

\begin{definition}
Let $\Lambda = (F,A,\varphi,Q)$ be a scenario. A \emph{linearization} of $\Lambda$ is a strict linear order of $F \cup A$ that extends $\prec_Q$. A ROBDD is said to \emph{represent} $\Lambda$ if it represents $\varphi$, and if its variable order is a linearization of $\prec_Q$.
\end{definition}

\begin{wrapfigure}[6]{r}{3cm}
\vspace{-7em}
\includegraphics[width=3cm]{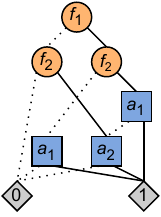}
\end{wrapfigure}

\begin{example} \label{ex:AFT3}
We continue Example \ref{ex:AFT2}. In this case $f_i \prec_Q a_j$ for all $i,j \in \{1,2\}$. We can extend this to a linear order by taking $f_1 < f_2 < a_1 < a_2$; the ROBDD corresponding to this order is pictured below, with filled lines connecting $v$ to $c_{\tone}(v)$ and dotted lines connecting $v$ to $c_{\tzero}(v)$.
\end{example}

\subsection{Relation to Markov decision processes}

A quantified scenario models a setting with probabilistic events, decisions, and costs associated to these decisions. A popular framework for modeling such settings are Markov decision processes (MDPs), which have a long history of quantitative analysis \cite{feinberg2012handbook}. In this section, we explain how the ROBDD of a quantified scenario can be viewed as an MDP and to what extent this aids us in calculating $\recht{PMC}(\hat{\Lambda})$ and $\recht{PEC}(\hat{\Lambda})$. We first recap the definition of an MDP.

\begin{definition}
A \emph{Markov decision process} is a tuple $(S,H,P,R)$ consisting of:
\begin{enumerate}
\item a finite set $S$ called the \emph{state space};
\item for each state $s \in S$ a finite set $H_s$ of \emph{actions};
\item for each two states $s,s' \in S$ and each action $h \in H_s$, a probability $P(s'|s,h) \in [0,1]$, such that $\sum_{s' \in S} P(s'|s,h) = 1$ for all $s \in S$ and $h \in H_s$;
\item for each two states $s,s' \in S$ and each action $h \in H_s$, a reward $R(s,h,s') \in \mathbb{R}$.
\end{enumerate}
A \emph{policy} $\sigma$ is the choice of a  $\sigma(s) \in H_s$ for each $s \in S$ with $H_s \neq \varnothing$.
\end{definition}

Given an MDP $(S,H,P,R)$ and a policy $\sigma$, an initial state $s_0$ defines a walk through $S$ as follows: at time $t$ and state $s_t$, one has $s_{t+1} = s'$ with probability $P(s'|s_t,\sigma(s_t))$. This process continues until a maximum time $t_{\max}$ is reached (which may be $\infty$), or until a $s$ with $H_s = \varnothing$ is reached. In other words, a choice of $\sigma$ turns $(S,H,P,R)$ into a discrete-time Markov chain.

MDPs are used to model a wide variety of systems in which decisions and probabilities both play a role. They are then analyzed to find a strategy that either satisfies a certain property or optimizes a certain objective function. Many types of analyses are possible, but two of the most important ones are:
\begin{description}
    \item[maximize reachability:] find the strategy that maximizes the probability that the walk will reach a given state.
    \item[maximize expected reward:] find the strategy $\sigma$ that maximizes the expected total reward $\mathbb{E}\left[\sum_{t=1}^{t_{\max}} R(s_{t-1},\sigma(s_{t-1}),s_t)\right]$.
\end{description}

As both probabilities and decisions play a role in quantified scenarios, it comes as no surprise that they can also be modeled as MDPs. In fact, the ROBDD of a quantified scenario can be considered an MDP: a vertex $v$ with its label in $F$ has only one action, with transitions to $c_{\tone}(v)$ and $c_{\tzero}(v)$ with probability $\pi_{t_V(v)}$ and $1-\pi_{t_V(v)}$, respectively. If $t_V(v) \in A$, then there are two actions at $v$, with transitions to $c_{\tone}(v)$ and $c_{\tzero}(v)$, respectively, both with probability 1. The transition $v \rightarrow c_{\tone}(v)$ also has reward $-\gamma_{t_V(v)}$, to denote the cost required to perform this action. This can be formalized into the following definition.

\begin{definition}
Let $\hat{\Lambda}$ be a quantified scenario. Its \emph{associated Markov decision process} is the MDP $(S,H,P,R)$ given by
\begin{align*}
S &= V,\\
H_v &= \begin{cases}
    \{\tzero,\tone\}, & \textrm{ if $t_V(v) \in A$},\\
    \{*\}, & \textrm{ if $t_V(v) \in F$},\\
    \varnothing, & \textrm{ if $t_V(v) \in \BB$},
\end{cases}\\
P(v'|v,h) &= \begin{cases}
 1, & \textrm{ if $t_V(v) \in A$ and $v' = c_h(v)$,}\\
\pi_{t_V(v)}, & \textrm{ if $t_V(v) \in F$ and $v'= c_{\tone}(v)$,}\\
1-\pi_{t_V(v)}, & \textrm{ if $t_V(v) \in F$ and $v'= c_{\tzero}(v)$,}\\
0, & \textrm{ otherwise,}
\end{cases}\\
R(v,h,v') &= \begin{cases}
-\gamma_{t_V(v)}, & \textrm{ if $t_V(v) \in A$, $h = \tone$, and $v' = c_{\tone}(v)$,}\\
0, & \textrm{ otherwise.}
\end{cases}
\end{align*}
\end{definition}

A policy for this MDP corresponds to a strategy for the scenario; however, the opposite is not true since a strategy's decision at a BDD node may depend on the path taken to get to that node, which is not allowed in MDPs. Nevertheless, our results in Section \ref{sec:BU} show that Pareto-optimal strategies are history-independent in the same sense. To calculate $\recht{PMC}(\hat{\Lambda})$ and $\recht{PEC}(\hat{\Lambda})$, we need techniques that differ from standard MDP techniques in several aspects:

\noindent\textbf{acyclicity:} The MDP is acyclic, so analysis can be simplified considerably \cite{bellman1957markovian}.

\noindent\textbf{multiple objectives:} Instead of optimizing a single objective, we either consider both probability (reachability of node $\tone$) and maximum/expected cost simultaneously. Multi-objective analysis of MDPs to return Pareto fronts has been studied in the literature \cite{etessami2008multi}.

\noindent\textbf{maximum cost:} MDP analysis usually considers the expected reward of a strategy, rather than the worst case probabilistic outcome. The worst case is more complicated, but algorithms exist as well \cite{de2005model}.

As can be seen from these references, the special properties of our setting have been individually tackled in the literature already, but to the best of my knowledge they have not been considered in unison before. In the next section, we provide an algorithm that finds the probability/cost Pareto front. As we will discuss there, it combines the ideas of the special cases handled in the literature.

\subsection{ROBDD analysis} \label{sec:BU}

In this section we show how we can calculate $\recht{PMC}(\hat{\Lambda})$ and $\recht{PEC}(\hat{\Lambda})$ using the ROBDD $B$. We start with $\recht{PMC}(\hat{\Lambda})$. Let $<$ be any linear order on $F \cup A$ extending $\prec_Q$, and suppose that the minimal element w.r.t. $<$ is $a \in A$. Let $\varphi^{\tzero},\varphi^{\tone}$ be the functions $\BB^F \times \BB^{A\setminus\{a\}} \rightarrow \BB$ obtained by substituting $\tzero$ (resp. $\tone$) for the $a$-argument of $\varphi$. We can consider $\varphi^{\tzero}$ as the function of a quantified scenario $\hat{\Lambda}^{\tzero}$, where $F^{\tzero},A^{\tzero},Q^{\tzero},\pi^{\tzero},\gamma^{\tzero}$ are all obtained from $\hat{\Lambda}$ by removing $a$ where appropriate. Any strategy $\sigma^{\tzero}$ of $\Lambda^{\tzero}$ can be extended to a strategy of $\sigma'^{\tzero}$ of $\Lambda$ by first choosing not to activate $a$, and then following $\sigma^{\tzero}$. Similarly, we can define $\hat{\Lambda}^{\tone}$, and extend a strategy $\sigma^{\tone}$ of $\Lambda^{\tone}$ to a strategy $\sigma'^{\tone}$ of $\Lambda$ by first activating $a$, and then following $\sigma^{\tone}$. Then each strategy for $\Lambda$ is of the form $\sigma'^{\tzero}$ or $\sigma'^{\tone}$, and
\begin{align*}
\tau_{\recht{m}}(\sigma'^{\tzero}) &= \tau_{\recht{m}}(\sigma^{\tzero}), & \tau_{\recht{m}}(\sigma'^{\tone}) &= \tau_{\recht{m}}(\sigma^{\tone})+ \vvec{0\\\pi(a)}.
\end{align*}
Thus $\tau_{\recht{m}}(\Sigma_{\Lambda}) = \tau_{\recht{m}}(\Sigma_{\Lambda^{\tzero}}) \cup \left(\tau_{\recht{m}}(\Sigma_{\Lambda^{\tone}})+\vvec{0\\\pi(a)}\right)$. When the minimal element w.r.t. $<$ is $f \in F$ instead, we get a similar expression of $\tau_{\recht{m}}(\Sigma_{\Lambda})$ in terms of $\tau_{\recht{m}}(\Sigma_{\Lambda^{\tzero}})$, $\tau_{\recht{m}}(\Sigma_{\Lambda^{\tone}})$, and $\pi_f$.

The key insight is now that in the ROBDD $B$, the sub-ROBDD with root $c_{\tzero}(\R{B})$ represents $\Lambda^{\tzero}$, while the sub-ROBDD with root $c_{\tone}(\R{B})$ represents $\Lambda^{\tone}$. We can leverage this to compute $\tau_{\recht{m}}(\Sigma_{\Lambda})$ bottom-up over the BDD. In fact, we can compute $\recht{PMC}(\hat{\Lambda})$ bottom-up, by discarding non-Pareto-optimal attacks at each step. This yields the following theorem, with its PEC counterpart below.

\begin{theorem} \label{thm:mainmax}
Let $\hat{\Lambda}$ be a quantified scenario. Let $B_{\Lambda} = (V,E,t_V,t_E)$ be an ROBDD representing $\Lambda$. For $v \in V$, define sets $G_v,\tilde{G}_v \subseteq \mathcal{D}$ recursively by
\begin{align*}
\tilde{G}_v &= \begin{cases}
\left\{\vvec{0\\0}\right\}, & \text{if $t_V(v) = \tzero$,}\\
\left\{\vvec{1\\0}\right\}, & \text{if $t_V(v) = \tone$,}\\
\left\{\vvec{\pi_f p_{\tzero}+(1-\pi_f)p_{\tone}\\\max(c_{\tzero},c_{\tone})} \ \middle| \ \vvec{p_{\tzero}\\c_{\tzero}} \in G_{c_{\tzero}(v)}, \vvec{p_{\tone}\\c_{\tone}} \in G_{c_{\tone}(v)}\right\}, & \textrm{if $t_V(v) = f \in F$,}\\
G_{c_{\tzero(v)}} \cup \left\{\vvec{p_{\tone}\\c_{\tone}+\gamma_a} \ \middle| \ \vvec{p_{\tone}\\c_{\tone}} \in G_{c_{\tone}(v)}\right\}, & \text{if $t_V(v) = a \in A$,}
\end{cases}\\
G_v &= \PF(\tilde{G}_v).
\end{align*}
Then $G_{\R{B}} = \recht{PMC}(\hat{\Lambda})$.
\end{theorem}

\begin{theorem} \label{thm:mainexp}
    Let $\hat{\Lambda}$ be a quantified scenario. Let $B_{\Lambda} = (V,E,t_V,t_E)$ be an ROBDD representing $\Lambda$. For $v \in V$, define sets $J_v,\tilde{J}_v \subseteq \mathcal{D}$ recursively by
\begin{align*}
\tilde{J}_v &= \begin{cases}
\left\{\vvec{0\\0}\right\}, & \text{if $t_N(v) = \tzero$,}\\
\left\{\vvec{1\\0}\right\}, & \text{if $t_N(v) = \tone$,}\\
\left\{\vvec{\pi_f p_{\tzero}+(1-\pi_f)p_{\tone}\\\pi_fc_{\tzero}+(1-\pi_f)c_{\tone}} \ \middle| \ \vvec{p_{\tzero}\\c_{\tzero}} \in J_{c_{\tzero}(v)}, \vvec{p_{\tone}\\c_{\tone}} \in J_{c_{\tone}(v)}\right\}, & \textrm{if $t_N(v) = f \in F$,},\\
J_{c_{\tzero(v)}} \cup \left\{\vvec{p_{\tone}\\c_{\tone}+\gamma_a} \ \middle| \ \vvec{p_{\tone}\\c_{\tone}} \in J_{c_{\tone}(v)}\right\}, & \text{if $t_N(v) = a \in A$,}
\end{cases}\\
J_v &= \SCPF(\tilde{J}_v).
\end{align*}
Then $J_{\R{B}} = \recht{PEC}(\hat{\Lambda})$.
\end{theorem}

These theorems are related to MDP analysis techniques as follows: in general, MDPs can be solved by a translation into a linear programming problem. When the MDP is acyclic, as in our case, the associated matrix is upper triangular, and the linear programming problem can be solved inductively. When multiple rewards are considered, this becomes a multi-objective linear programming problem \cite{etessami2008multi}, of which the output is a Pareto front rather than a single value; this can still be computed inductively. Furthermore, the analysis of maximum cost, rather than expected cost, is generally more complicated as it is no longer history-independent \cite{de2005model}. In general, this can be tackled by increasing the number of variables of the optimization problem: in our case, our theorems show that this is not necessary and the same bottom-up approach still works for PMC.

\begin{example} \label{ex:AFT4}
We continue Example \ref{ex:AFT3}, with $\pi$ and $\gamma$ from Example \ref{ex:AFT2}. We calculate the $\tilde{G}_v$ bottom-up; for convenience, we identify nodes with their labels, using \emph{left}, \emph{right} to disambiguate when two nodes have the same label. Underlined vectors are not Pareto-optimal and discarded when going from $\tilde{G}_v$ to $G_v$.
\begin{align*}
\tilde{G}_{\tzero} &= \left\{\vvec{0\\0}\right\}, & \tilde{G}_{\tone} &= \left\{\vvec{1\\0}\right\},\\
\tilde{G}_{a_1,\recht{left}} = \tilde{G}_{a_2} &= \left\{\vvec{0\\0},\vvec{1\\10}\right\}, & \tilde{G}_{a_1,\recht{right}} &= \left\{\vvec{0\\0},\vvec{1\\10}\right\},\\
\tilde{G}_{f_2,\recht{left}} &= \left\{\vvec{0\\0},\vvec{0.5\\10}\right\}, &
\tilde{G}_{f_2,\recht{right}} &= \left\{\vvec{0\\0},\underline{\vvec{0.5\\10}},\vvec{1\\10}\right\},\\
\tilde{G}_{f_1} &= \left\{\vvec{0\\0},\underline{\vvec{0.25\\10}},\underline{\vvec{0.5\\10}},\vvec{0.75\\10}\right\}.
\end{align*}
Thus we obtain the same Pareto front $\left\{\vvec{0\\0},\vvec{0.75\\10}\right\}$ of Section \ref{sec:example}.
\end{example}

\begin{figure}
    \centering
    \includegraphics[width=10cm]{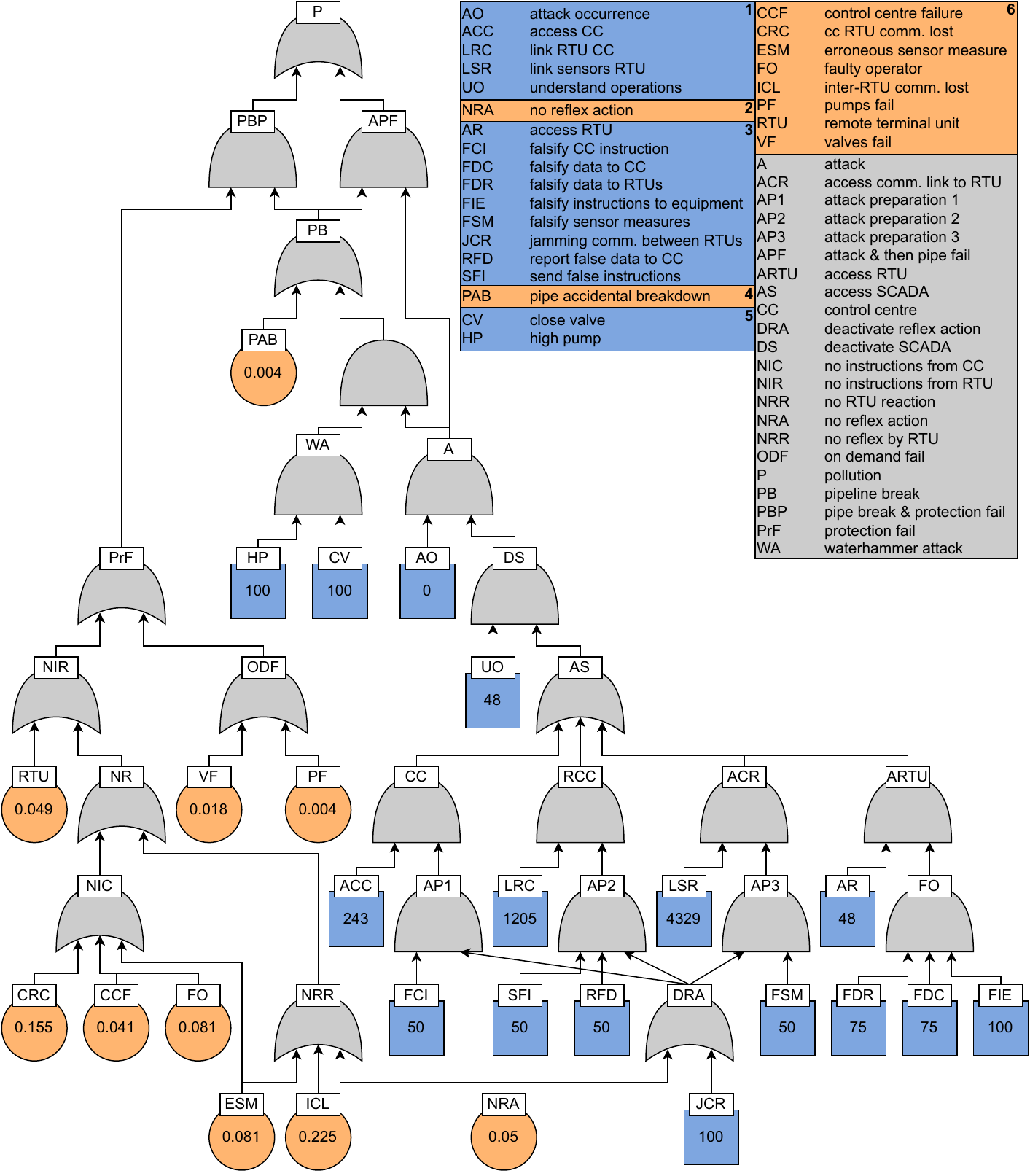}
    \caption{Attack-fault tree for an oil pipe line, adapted from \cite{kumar2017quantitative}. Orange circles denote BCFs with inscribed failure probability, while blue squares denote BASs with inscribed (time) cost. The partial order $\prec_Q$ is given by the blocks on the top left, where the coloured blocks are assumed to happen simultaneously with block $i$ before block $i+1$.}
    \label{fig:OPL}
\end{figure}

\section{Case study: oil pipeline}

We validate our methods by applying them to an oil pipe line case study (Fig.~\ref{fig:OPL}). The AFT is adapted from \cite{kumar2017quantitative}, which in turn is adapted from a BDMP model from \cite{kriaa2014safety}. The oil pipe can leak polluting material either through the combination of a cyberattack and a pipeline failure (\texttt{APF}) or pipeline failure, followed by the failure of the protection mechanism (\texttt{PBP}). A cyberattack targets the Supervisory Control And Data Acquisition (SCADA) system through the Remote Terminal Units (RTU) and Control Centre (CC); see \cite{kumar2017quantitative} for more details. The AFT in \cite{kumar2017quantitative} is dynamic: it contains dynamic attack and failure gates, and probabilities are given by exponential distributions. We use the dynamic gates to determine the order $\prec_Q$, and we convert the exponential distributions for BCFs to fixed probabilities by taking a 1-year time horizon. Furthermore, we measure attacker cost in the expected attack time, taking the expected value for each BAS with a given attack rate. For attacks with a detection probability, we instead assume that a BAS is reattempted 30 days later when detected.

To calculate PMC and PEC for this AFT, we first construct the BDD, which (ordering BCFs/BASs within each block of Fig.~\ref{fig:OPL} alphabetically) has 66 nodes. Applying Theorems \ref{thm:mainmax} and \ref{thm:mainexp}, we get
\begin{align*}
\recht{PMC}(\hat{\Lambda}) &= \left\{\vvec{0.0021\\0},\vvec{0.004\\346},\vvec{0.0538\\541}, \vvec{1\\546} \right\},\\
\recht{PEC}(\hat{\Lambda}) &= \left\{\vvec{0.0021\\0},\vvec{1\\545.2} \right\}.
\end{align*}
We can draw a number of conclusions. First, for PMC, of the $2^{31}$ possible strategies, only 4 are Pareto-optimal. Second, the attacker needs to spend considerable cost to affect the compromise probability. Third, the cost difference between the last two Pareto-optimal attacks is minimal compared to the probability difference; this suggest that an attacker will usually go for the last option, unless they have a hard cost constraint. This is also reflected in the PEC, where only the last attack is Pareto-optimal. The bottom-up calculation can be adapted to store the corresponding strategy; for $\vvec{0.0040\\346}$ this is the attack consisting of \texttt{AO}, \texttt{UO}, \texttt{AR}, \texttt{FDR}, \texttt{FDC}, \texttt{FIE}, which compromises the AFT when the BCF \texttt{PAB} occurs.

As discussed in Section \ref{ssec:compare}, no other work on AFTs is able to compute the cost-probability Pareto front, so we cannot do a direct computation time comparison. However, we can say that our state space of 66 BDD nodes is quite small compared to the size of the AFT (50 nodes); the analysis of this BDD only takes 80ms. By contrast, an automata-based approach like \cite{kumar2017quantitative,andre2019parametric} would involve creating automata for all nodes, synchronized via channels. The total number of locations for these automata will be a multiple of the number of nodes, so the total state space, obtained by composing the node automata, will be significantly larger, which would also result in a significantly longer computation time. 

\section{Conclusion}

We have introduced a novel method to quantitatively explore the interplay between safety and security via attack-fault trees, by calculating the Pareto front between failure probability and attacker cost using BDDs. This approach is the first full Pareto analysis on AFTs, and is significantly more lightweight than existing automata-based analysis methods. It also comes with rigorous definitions and proofs.

This work can be extended in several ways. First, right now the BDD from the AFT is obtained by taking a random linear order that extends $\prec_Q$. By adapting variable reordering heuristics \cite{jiang2017variable} to respect this partial order, we can get smaller BDDs representing the same AFT, thus reducing computation time. Second, in FT/AT analysis parameter values are not always known exactly. Incorporating uncertainty into the analysis requires tools from interval MDP analysis \cite{hahn2019interval}.

\bibliographystyle{splncs04}
\bibliography{refs.bib}

\appendix

\section{Proof of Lemma \ref{lem:spo}}

Irreflexivity and asymmetry are immediately clear, so what is left is transitivity: we need to prove that if $x \prec_{\Lambda} y$ and $y \prec_{\Lambda} z$, then $x \prec_{\Lambda} Z$ for all $x,y,z \in F \cup A$. This proof is not hard, but there are many cases to cover:
\begin{itemize}
\item $x,y,z \in F$: by assumption, there exist $a_1,a_2$ such that $x \in Q_{a_1}$, $y \notin Q_{a_1}$, $y \in Q_{a_2}$, $z \notin Q_{a_2}$. The $Q_a$ are linearly ordered by inclusion, so $Q_{a_1} \subset Q_{a_2}$ as the latter includes $b$ and the former does not. Hence $x \in Q_{a_2} \not \ni z$, so $x \prec_{\Lambda} z$.
\item $x,y \in F$, $z \in A$: By assumption, there exists a $a$ such that $x \in Q_a \not \in y$; furthermore, we know that $y \in Q_z$. Since the $Q_{a'}$ are linearly ordered by $\subseteq$, we get $Q_a \subset Q_z$. Hence $x \in Q_z$ and $x \prec_{\Lambda} z$.
\item $x,z \in F$, $y \in A$: We have $x \in Q_y \not \ni z$, so $x \prec_{\Lambda} z$.
\item $y,z \in F$, $x \in A$: By assumption $y \notin Q_x$, and there exists a $a\in A$ such that $y \in Q_a \not \ni z$. Hence $Q_x \subset Q_a$, so $z \notin Q_x$ and $x \prec_{\Lambda} z$.
\item $x \in F$, $x,y \in A$: We have $a \in Q_b \subset Q_c$, so $a \prec_{\Lambda} c$.
\item $y \in F$, $x,z \in A$: We have $Q_x \not \ni y \in Q_z$. Hence $Q_x \subset Q_z$, so $x \prec_{\Lambda} z$.
\item $z \in F$, $x,y \in A$: We have $Q_x \subset Q_y \not \ni z$, so $z \notin Q_x$, so $x \prec_{\Lambda} z$.
\item $x,y,z \in A$: We have $Q_x \subset Q_y \subset Q_z$, so $x \prec_{\Lambda} z$.
\end{itemize}

\section{Proof of Theorem \ref{thm:mainmax}}

In this section we prove Theorem \ref{thm:mainmax}. The proof of Theorem \ref{thm:mainexp} is completely analogous, so we will not present this. The proof is quite involved, but it consists of the following steps:
\begin{description}
\item[\ref{ssec:commute}:] We show that the Pareto front operator $\PF$ commutes with the operations used to create $\tilde{G}_v$ from $G_{c_{\tzero}(v)}$ and $G_{c_{\tone}(v)}$. This reduces our problem of computing $\recht{PMC}(\hat{\Lambda})$ to computing $\tau_{\recht{m}}(\Sigma_{\Lambda})$ bottom-up.
\item[\ref{ssec:subdiv}:] We formally state and prove the assertions about the derived scenarios $\hat{\Lambda}^{\tzero}$, $\hat{\Lambda}^{\tone}$ of Section \ref{sec:BU}. This allows us to propagate strategies upward in the BDD. 
\item[\ref{ssec:fobdd}:] We use these results to prove Theorem \ref{thm:mainmax} for the full ordered BDD (FOBDD), the BDD representing $\varphi$ of maximal size.
\item[\ref{ssec:reduce}] Finally, we show how reducing the FOBDD to the ROBDD does not affect calculation, thus proving Theorem \ref{thm:mainmax}.
\end{description}

\subsection{Operations on $\mathcal{D}$} \label{ssec:commute}

In this subsection, we prove that the operations on subsets of $\mathcal{D}$ used in Theorem \ref{thm:mainmax} commute with the Pareto front operator in a suitable sense. The main result is Propositions \ref{prop:commutef} and \ref{prop:commutea}. We first define the following operators:
\begin{align*}
\psi\left(\vvec{p_1\\c_1},\vvec{p_2\\c_2},p\right) &= \vvec{(1-p)p_1+pp_2\\\max(c_1,c_2)} && \text{for $\vvec{p_1\\c_1},\vvec{p_2\\c_2} \in \mathcal{D}$ and $p \in [0,1]$,}\\
\Psi_{\recht{F}}(I_1,I_2,p) &= \{\psi(d_1,d_2,p) \mid d_1 \in I_1, d_2 \in I_2\}&& \text{for $I_1, I_2 \subseteq \mathcal{D}$ and $p \in [0,1]$,}\\
\Psi_{\recht{D}}(I_1,I_2,c) &= I_1 \cup \left\{\vvec{p_2\\c_2+c} \ \middle|\ \vvec{p_2\\c_2} \in I_2\right\}&& \text{for $I_1, I_2 \subseteq \mathcal{D}$ and $a \in [0,\infty]$.}
\end{align*}

Note that $\psi$ is monotonous w.r.t. $\sqsubseteq$ in its first two arguments.

\begin{lemma} \label{lem:pfcap}
If $I_1 \subseteq I_2 \subseteq \mathcal{D}$, then $\PF(I_2) \cap I_1 \subseteq \PF(I_1)$.
\end{lemma}

\begin{proof}
Let $d \in \PF(I_2) \cap I_1$. Since $d \in I_1$, there exists a $d' \in \PF(I_1)$ such that $d' \sqsubseteq d$. Since $d' \in I_2$, there exists a $d'' \in \PF(I_2)$ such that $d'' \sqsubseteq d' \sqsubseteq d$. Since both $d$ and $d''$ are Pareto optimal in $I_2$, these $\sqsubseteq$ must be equalities, so $d = d' \in \PF(I_1)$.
\end{proof}

\begin{proposition} \label{prop:commutef}
For $I_1,I_2 \subseteq \mathcal{D}$ and $p \in [0,1]$ one has $\PF(\Psi(\PF(I_1),I_2,p)) = \PF(\Psi(I_1,I_2,p)) = \PF(\Psi(I_1,\PF(I_2),p))$.
\end{proposition}

\begin{proof}
We only prove the first equality as the other is analogous. Let $d_1 \in \PF(I_1)$, $d_2 \in I_2$ be such that $\psi(d_1,d_2,p) \in \PF(\Psi_{\recht{F}}(\PF(I_1),I_2,p))$. Let $d_1' \in I_1$, $d_2' \in I_2$ be such that $\psi(d_1',d_2',p) \sqsubseteq \psi(d_1,d_2,p)$. Let $d_1'' \in \PF(I_1)$ be such that $d_1'' \sqsubseteq d_1'$; then because $\psi$ is monotonous we get
\[
\psi(d_1'',d_2',p) \sqsubseteq \psi(d_1',d_2',p) \sqsubseteq \psi(d_1,d_2,p).\]
Since $\psi(d_1'',d_2',p) \in \Psi_{\recht{F}}(\PF(I_1),I_2,p)$ and $\psi(d_1,d_2,p)$ was assumed to be Pareto-optimal in this set, both $\sqsubseteq$ above must be equalities. In particular the case $\psi(d_1',d_2',p) \sqsubset \psi(d_1,d_2,p)$ is impossible, which shows $\psi(d_1,d_2,p) \in \PF(\Psi_{\recht{F}}(I_1,I_2,p))$. This shows ``$\subseteq$''.

Now let $d_1 \in I_1$, $d_2 \in I_2$ be such that $\psi(d_1,d_2,p) \in \PF(\Psi_{\recht{F}}(I_1,I_2,p))$. Let $d_1' \in \PF(I_1)$ such that $d_1' \sqsubseteq d_1$. Then $\psi(d_1',d_2,p) \sqsubseteq \psi(d_1,d_2,p)$, but since the latter was assumed to be Pareto optimal, this is an equality; hence by replacing $d_1$ by $d_1'$ if necessary, we may assume that $d_1 \in \PF(I_1)$, so $\psi(d_1,d_2,p) \in \Psi_{\recht{F}}(\PF(I_1),I_2,p)$. Applying Lemma \ref{lem:pfcap} to $\Psi_{\recht{F}}(I_1,I_2,p) \supseteq \Psi_{\recht{F}}(\PF(I_1),I_2,p)$, we find $\psi(d_1,d_2,p) \in \PF(\Psi_{\recht{F}}(\PF(I_1),I_2,p))$. This shows ``$\supseteq$'' and completes the proof.
\end{proof}

\begin{lemma} \label{lem:pfcup}
For $I_1,I_2 \subseteq \mathcal{D}$ one has $\PF(I_1 \cup I_2) = \PF(\PF(I_1) \cup I_2)$.
\end{lemma}

\begin{proof}
The proof consists of the following four statements:
\begin{enumerate}
    \item \emph{If $d \in I_1 \cap \PF(I_1 \cup I_2)$, then $d \in \PF(\PF(I_1) \cup I_2)$:} Applying Lemma \ref{lem:pfcap} on $I_1 \subseteq I_1 \cup I_2$ we get $d \in \PF(I_1)$, so $d \in \PF(I_1) \cup I_2$. Applying Lemma \ref{lem:pfcap} again on $\PF(I_1) \cup I_2 \subseteq I_1 \cup I_2$, we get $d \in \PF(\PF(I_1) \cup I_2)$.
    \item \emph{If $d \in I_2 \cap \PF(I_1 \cup I_2)$, then $d \in \PF(\PF(I_1) \cup I_2)$:} Apply Lemma \ref{lem:pfcap} to $\PF(I_1) \cup I_2 \subseteq I_1 \cup I_2$.
    \item \emph{If $d \in \PF(I_1) \cap \PF(\PF(I_1) \cup I_2)$, then $d \in \PF(I_1 \cup I_2)$:} Let $d' \in \PF(I_1 \cup I_2)$ be such that $d' \sqsubseteq d$. If $d' \in I_1$, then $d'= d$ as $d$ is Pareto optimal in $I_1$. If $d' \in I_2$, then $d' = d$ as $d$ is Pareto optimal in $\PF(I_1) \cup I_2$. Either way we find $d \in \PF(I_1 \cup I_2)$.
    \item \emph{If $d \in I_2 \cap \PF(\PF(I_1) \cup I_2)$, then $d \in \PF(I_1 \cup I_2)$:} Let $d' \in \PF(I_1 \cup I_2)$ be such that $d' \sqsubseteq d$. If $d' \in I_1$, then let $d'' \in \PF(I_1)$ be such that $d'' \sqsubseteq d'$. Then $d'' \subseteq d' \sqsubseteq d$, but since $d$ is Pareto optimal in $\PF(I_1) \cup I_2$, we find $d'' = d$, hence $d = d'$. If $d' \in I_2$, we also find $d = d'$ since $d$ is Pareto-optimal in $\PF(I_1) \cup I_2$. Hence we conclude that $d \in \PF(I_1 \cup I_2)$.
\end{enumerate}
\end{proof}

\begin{proposition} \label{prop:commutea}
For $I_1, I_2 \subseteq \mathcal{D}$ and $c \in [0,\infty]$ one has $\PF(\Psi_{\recht{A}}(I_1,I_2,c)) = \PF(\Psi_{\recht{A}}(\PF(I_1),I_2,c)) = \PF(\Psi_{\recht{A}}(I_1,\PF(I_2),c))$.
\end{proposition}

\begin{proof}
Since addition is monotonous, one has 
\[
    \PF\left\{\vvec{p_2\\c_2+c} \ \middle|\ \vvec{p_2\\c_2} \in I_2\right\} = \left\{\vvec{p_2\\c_2+c} \ \middle|\ \vvec{p_2\\c_2} \in \PF(I_2)\right\}.
\]
Using Lemma \ref{lem:pfcup} we then get
\begin{align*}
\PF(\Psi_{\recht{A}}(I_1,I_2,c)) &= \PF\left(I_1 \cup \left\{\vvec{p_2\\c_2+c} \ \middle|\ \vvec{p_2\\c_2} \in I_2\right\} \right)\\
&= \PF\left(I_1 \cup \PF\left\{\vvec{p_2\\c_2+c} \ \middle|\ \vvec{p_2\\c_2} \in I_2\right\} \right)\\
&= \PF\left(I_1 \cup \left\{\vvec{p_2\\c_2+c} \ \middle|\ \vvec{p_2\\c_2} \in \PF(I_2)\right\} \right)\\
&= \PF(\Psi_{\recht{A}}(I_1,\PF(I_2),c)).
\end{align*}
The proof that $\PF(\Psi_{\recht{A}}(I_1,I_2,c)) = \PF(\Psi_{\recht{A}}(\PF(I_1),I_2,c))$ is analogous.
\end{proof}

\subsection{Scenario reductions} \label{ssec:subdiv}

In this section we create, from a scenario $\Lambda = (F,A,\varphi,Q)$, new reduced scenarios $\Lambda^{\tzero},\Lambda^{\tone}$ by substituting $\tzero$ or $\tone$ for the lowest-ranked variable in $\varphi$ (according to $\prec_Q$). It turns out that the strategies of $\Lambda$, and their probabilities and costs, can be expressed in terms of the strategies of $\Lambda^{\tzero}$ and $\Lambda^{\tone}$. This forms the basis of the induction proof of Theorem \ref{thm:mainmax}. The main results of this subsection are Propositions \ref{prop:tauf} and \ref{prop:taua}. We first define reduced scenarios. For a Boolean function $\varphi\colon \BB^F \times \BB^A \rightarrow \BB$, and $f \in F$ we write $\varphi[x_f/\tzero]$ for the function $\BB^{F\setminus\{f\}} \times \BB^A \rightarrow \BB$ obtained by substituting $\tzero$ for the $f$-labeled variable in $\varphi$. We likewise define $\varphi[x_f/\tone]$ and, for $a \in A$, $\varphi[y_a/\tzero]$ and $\varphi[y_a/\tone]$.

\begin{definition}
Let $\hat{\Lambda} = (F,A,\varphi,Q,\pi,\gamma)$ be a quantified scenario.
\begin{enumerate}
\item If there exists a $f \in F$ that is minimal w.r.t. $\prec_Q$ (i.e., $f \in Q_a$ for all $a$), then we define the quantified scenarios $\hat{\Lambda}^{f/\tzero}, \hat{\Lambda}^{f/\tone}$ by
\begin{align*}
F^{f/\tzero} = F^{f/\tone} &= F \setminus \{f\},& A^{f/\tzero} = A^{f/\tone} &= A,\\
\varphi^{f/\tzero} &= \varphi[x_f/\tzero], & \varphi^{f/\tone} &= \varphi[x_f/\tone],\\
Q^{f/\tzero} = Q^{f/\tone} &= (a \mapsto (Q_a \setminus \{f\})), & \pi^{f/\tzero} = \pi^{f/\tone} &= \pi|_{F\setminus \{f\}},\\
\gamma^{f/\tzero} = \gamma^{f/\tone} &= \gamma.
\end{align*}
\item If there exists a $a \in A$ that is minimal w.r.t. $\prec_Q$ (i.e., if $Q_a = \varnothing$), then define the quantified scenarios $\hat{\Lambda}^{a/\tzero},\hat{\Lambda}^{a/\tone}$ by
\begin{align*}
F^{a/\tzero} = F^{a/\tone} &= F,& A^{a/\tzero} = A^{a/\tone} &= A\setminus\{a\},\\
\varphi^{a/\tzero} &= \varphi[y_a/\tzero], & \varphi^{a/\tone} &= \varphi[y_a/\tone],\\
Q^{a/\tzero} = Q^{a/\tone} &= (a' \mapsto Q_{a'}), & \pi^{a/\tzero} = \pi^{a/\tone} &= \pi,\\
\gamma^{a/\tzero} = \gamma^{a/\tone} &= \gamma|_{A\setminus \{a\}}.
\end{align*}
\end{enumerate}
\end{definition}

Assume there is a minimal $f \in F$ w.r.t. $\prec_Q$. We now discuss how (the probabilities and costs of) the strategies of $\Lambda^{f/\tzero},\Lambda^{f/\tone}$ relate to those of $\Lambda$. Note that if we have strategies $\sigma^{\tzero}$ for $\Lambda^{f/\tzero}$ and $\sigma^{\tone}$ for $\Lambda^{f/\tone}$, then we can compose them into a strategy $\sigma^{\tzero} \star \sigma^{\tone}$ for $\Lambda$, by first observing $\mathsf{x}_f$, and then following $\sigma^{\tzero}$ if $\mathsf{x}_f = \tzero$, and $\sigma^{\tone}$ if $\mathsf{x}_f = \tone$. This is possible since we assume $f$ to be minimal w.r.t. $\prec_Q$. The following result states that every strategy of $\Lambda$ is obtained in this way.

\begin{lemma}
Let $\Lambda = (F,A,\varphi,Q)$ be a scenario, and suppose there exists a $f \in F$ that is minimal in $F \cup A$ w.r.t. $\prec_Q$. Then there exists a bijection $\star\colon \Sigma_{\Lambda^{f/\tzero}} \times \Sigma^{\Lambda_{f/\tone}} \rightarrow \Sigma_{\Lambda}$ defined, for all $a \in A$ and $x \in \BB^{Q_a}$, by
\[
(\sigma^{\tzero} \star \sigma^{\tone})_a(x) = (\neg x_f \wedge \sigma^0_a(x')) \vee (x_f \wedge \sigma^{\tone}_a(x')),
\]
where $x' \in \BB^{Q_{a} \setminus \{f\}}$ satisfies $x'_{f'} = x_{f'}$ for all $f' \in Q_a \setminus \{f\}$.
\end{lemma}

\begin{proof}
The inverse is given by $\sigma \mapsto (\sigma_a[x_f/\tzero],\sigma_a[x_f/\tone])_{a \in A}$.
\end{proof}

Our next result shows that the probability and cost of $\sigma^{\tzero} \star \sigma^{\tone}$ can be expressed in those of $\sigma^{\tzero}$ and $\sigma^{\tone}$.

\begin{lemma}
    For $\sigma^{\tzero} \in \Sigma_{\Lambda^{f/\tzero}}, \sigma^{\tone} \in \Sigma^{\Lambda_{f/\tone}}$ one has
    \begin{align*}
    \hat{\pi}(\sigma^{\tzero} \star \sigma^{\tone}) &= (1-\pi_f)\hat{\pi}^{f/\tzero}(\sigma^{\tzero})+\pi_f\hat{\pi}^{f/\tone}(\sigma^{\tone}),\\
    \hat{\gamma}_{\recht{m}}(\sigma^{\tzero} \star \sigma^{\tone}) &= \max(\hat{\gamma}_{\recht{m}}^{f/\tzero}(\sigma^{\tzero}),\hat{\gamma}_{\recht{m}}^{f/\tone}(\sigma^{\tone})).
    \end{align*}
    \end{lemma}
    
    \begin{proof}
    Let $\sigma = \sigma^{\tzero} \star \sigma^{\tone}$. Then for $x \in X$ we have
    \[
    \hat{\sigma}(x) = \begin{cases}
        \hat{\sigma}^{\tzero}(x'), & \textrm{ if $x_f = \tzero$,}\\
        \hat{\sigma}^{\tone}(x'), & \textrm{ if $x_f = \tone$.}
    \end{cases}
    \]
    Hence
    \begin{align*}
    \hat{\pi}(\sigma) &= \prob(\varphi(\mathsf{x},\hat{\sigma}(\mathsf{x})) = \tone) \\
    &= \prob(\mathsf{x}_f = \tzero)\prob(\varphi(\mathsf{x},\hat{\sigma}(\mathsf{x})) = \tone|\mathsf{x}_f = \tzero)+\prob(\mathsf{x}_f = \tone)\prob(\varphi(\mathsf{x},\hat{\sigma}(\mathsf{x})) = \tone|\mathsf{x}_f = \tone)\\
    &= (1-\pi_f)\prob(\varphi(\mathsf{x},\hat{\sigma}(\mathsf{x})) = \tone|\mathsf{x}_f = \tzero)+\pi_f\prob(\varphi(\mathsf{x},\hat{\sigma}(\mathsf{x})) = \tone|\mathsf{x}_f = \tone)\\
    &= (1-\pi_f)\prob(\varphi^{f/\tzero}(\mathsf{x}',\hat{\sigma}_{\tzero}(\mathsf{x}')) = \tone)+\pi_f\prob(\varphi^{f/\tone}(\mathsf{x}',\hat{\sigma}_{\tone}(\mathsf{x}')) = \tone)\\
    &= (1-\pi_f)\hat{\pi}^{f/\tzero}(\sigma_{\tzero})+\pi_f\hat{\pi}^{f/\tone}(\sigma_{\tone}).
    \end{align*}
    Furthermore
    \begin{align*}
    \hat{\gamma}_{\recht{m}}(\sigma) &= \max_{x \in X} \sum_{a \in A} \gamma_a\hat{\sigma}(x)_a \\
    &= \max\left(\max_{x' \in X'} \sum_{a \in A} \gamma_a\hat{\sigma}^{\tzero}(x')_a,\max_{x' \in X'} \sum_{a \in A} \gamma_a\hat{\sigma}^{\tone}(x')_a\right)\\
    &= \max(\hat{\gamma}_{\recht{m}}^{f/\tzero}(\sigma^{\tzero}),\hat{\gamma}_{\recht{m}}^{f/\tone}(\sigma^{\tone})).
    \end{align*}
    \end{proof}

Combining these two lemmas leads to the following result, where $\Psi_{\recht{F}}$ is as in Section \ref{ssec:commute}:

\begin{proposition} \label{prop:tauf}
Let $hat{\Lambda}$ be a quantified scenario such that there exists an $f \in F$ that is minimal in $F \cup A$ w.r.t. $\prec_Q$. Then
\[
\tau_{\recht{m}}(\Sigma_{\Lambda}) = \Psi_{\recht{F}}(\Sigma_{\Lambda^{f/\tzero}},\Sigma_{\Lambda^{f/\tone}},\pi_f).
\]
\end{proposition}

We now turn our attention to the case that the minimal element of $F \cup A$ is $a \in A$. In this case, we do not combine strategies $\sigma^{\tzero}$ of $\Lambda^{a/\tzero}$ and $\sigma^{\tone}$ of $\Lambda^{a/\tone}$, but rather lift them separately: we get a strategy $\iota(\sigma^{\tzero})$ by not performing $a$ and then following $\sigma^{\tzero}$, and a strategy $\iota(\sigma^{\tone})$ by performing $a$ and then following $\sigma^{\tzero}$. The following result shows that all strategies of $\Lambda$ are obtained in this way. Note that as sets, $\Sigma_{\Lambda^{a/\tzero}}$ and $\Sigma_{\Lambda^{a/\tone}}$ are equal, since $\Lambda^{a/\tzero}$ and $\Lambda^{a/\tone}$ only differ in their Boolean functions. We therefore use their disjoint union, denoted by $\sqcup$, to ensure that their elements are treated as distinct.

\begin{lemma}
Let $\Lambda = (F,A,\varphi,Q)$ be a scenario such that there exists an $a \in A$ that is minimal in $F \cup A$ w.r.t. $\prec_Q$. Then there exists a bijection $\iota \colon \Sigma_{\Lambda^{a/\tzero}} \sqcup\Sigma_{\Lambda^{a/\tone}} \rightarrow \Sigma_{\Lambda}$ defined, for all $a' \in A$ by
\begin{align*}
\iota(\sigma)_{a'} &= \begin{cases}
\sigma_{a'},& \textrm{ if $a' \neq a$,}\\
\tzero,& \textrm{ if $a' = a$ and $\sigma \in \Sigma_{\Lambda^{a/\tzero}}$,}\\
\tone,& \textrm{ if $a' = a$ and $\sigma \in \Sigma_{\Lambda^{a/\tone}}$.}
\end{cases}
\end{align*}
\end{lemma}

\begin{proof}
The inverse is given by
\[
\sigma \mapsto \begin{cases}
(\sigma)|_{A \setminus \{a\}} \in \Sigma_{\Lambda_{a/\tzero}}, & \textrm{ if $\sigma_a = \tzero$,}\\
(\sigma)|_{A \setminus \{a\}} \in \Sigma_{\Lambda_{a/\tone}}, & \textrm{ if $\sigma_a = \tone$.}
\end{cases}
\]
\end{proof}

As before, we can express the probability and cost of strategies of $\Lambda$ in terms of those of strategies of $\Lambda^{a/\tzero}$ and $\Lambda^{a/\tone}$.

\begin{lemma}
For $\sigma^{\tzero} \in \Sigma_{\Lambda^{a/\tzero}}, \sigma^{\tone} \in \Sigma_{\Lambda^{a/\tone}}$ one has
\begin{align*}
\hat{\pi}(\iota(\sigma^{\tzero})) &= \hat{\pi}^{a/\tzero}(\sigma^{\tzero}),&
\hat{\pi}(\iota(\sigma^{\tone})) &= \hat{\pi}^{a/\tone}(\sigma^{\tone}),\\
\hat{\gamma}_{\recht{m}}(\iota(\sigma^{\tzero})) &= \hat{\gamma}^{a/\tzero}_{\recht{m}}(\sigma^{\tzero}),&
\hat{\gamma}_{\recht{m}}(\iota(\sigma^{\tone})) &= \hat{\gamma}^{a/\tone}_{\recht{m}}(\sigma^{\tone})+\gamma_a.
\end{align*}
\end{lemma}

\begin{proof}
We prove this for $\sigma_{\tone}$ only, as the other case is analogous. One has $\widehat{\iota(\sigma_{\tone})}_a = \tone$ and $\widehat{\iota(\sigma_{\tone})}_{a'} = (\hat{\sigma}_{\tone})_{a'}$ for $a' \neq a$. Hence
\begin{align*}
\hat{\pi}(\iota(\sigma^{\tone})) &= \prob(\varphi(\mathsf{x},\widehat{\iota(\sigma_{\tone})}(\mathsf{x})) = \tone)& \hat{\gamma}_{\recht{m}}(\iota(\sigma^{\tone})) &= \max_{x \in X} \sum_{a' \in A} \gamma_{a'} \widehat{\iota(\sigma^{\tone})}(x)_{a'} \\
&= \prob(\varphi^{a/\tone}(\mathsf{x},\hat{\sigma}^{\tone}(\mathsf{x})) = \tone) &&= \gamma_a + \max_{x \in X} \sum_{a' \neq a} \gamma_{a'} \hat{\sigma}^{\tone}(x)_{a'} \\
&= \hat{\pi}^{a/\tone}(\sigma^{\tone}) &&= \gamma_a + \hat{\gamma}_{\recht{m}}^{a/\tone}(\sigma^{\tone}).
\end{align*}
\end{proof}

Combining these two lemmas leads to the following result, where $\Psi_{\recht{A}}$ is as in Section \ref{ssec:commute}:

\begin{proposition} \label{prop:taua}
    Let $hat{\Lambda}$ be a quantified scenario such that there exists an $a \in A$ that is minimal in $F \cup A$ w.r.t. $\prec_Q$. Then
    \[
    \tau_{\recht{m}}(\Sigma_{\Lambda}) = \Psi_{\recht{A}}(\Sigma_{\Lambda^{a/\tzero}},\Sigma_{\Lambda^{a/\tone}},\gamma_a).
    \]
    \end{proposition}

\subsection{Theorem \ref{thm:mainmax} for FOBDDs} \label{ssec:fobdd}

In this section, we generalize the notion of ROBDDs to ordered binary decision diagram (OBDD). There are multiple OBDDs representing a single Boolean functions, with the ROBDD being the minimal one. The maximal one is called the full OBDD (FOBDD); in this section we prove Theorem \ref{thm:mainmax} for FOBDDs. In the next section, we prove that the reduction steps to get from the FOBDD to the ROBDD do not affect the result, thus proving Theorem \ref{thm:mainmax}. We start with the definition of (F)OBDD.

\begin{definition}
Let $C$ be a finite set, and let $<$ be a strict linear order on $C$. An \emph{ordered binary decision diagram} (OBDD) on $(C,<)$ is a tuple $B = (V,E,t_V,t_E)$ where:
\begin{enumerate}
\item $(V,E)$ is a rooted directed acyclic graph.
\item $t_V\colon V \rightarrow C \cup \BB$ satisfies $t_V(v) \in \BB$ iff $v$ is a leaf.
\item $t_E\colon E \rightarrow \{\tzero,\tone\}$ is such that the two outgoing edges of each nonleaf $v$ have different values; the two children are denoted $c_{\tzero}(v)$ and $c_{\tone}(v)$ depending on the edge values.
\item If $(v,v') \in E$, then $t_V(v) < t_V(v')$.
\end{enumerate}
An OBDD is called \emph{full} (FOBDD) if it is a tree and each leaf has depth $|C|$. 
\end{definition}

An (F)OBDD represents a Boolean function in the same manner as an ROBDD, though the representation is no longer unique. The following result shows that all OBDDs representing the same function are related through reduction steps:

\begin{proposition} \emph{\cite{bryant1986graph}} \label{prop:bdd}
Let $C$ be a finite set, let $<$ be a linear order on $C$, and let $\varphi$ be a Boolean function $\varphi\colon \BB^C \rightarrow \BB$.
\begin{enumerate}
    \item There is a unique FOBDD on $(C,<)$ representing $\varphi$.
    \item Any OBDD on $(C,<)$ can be obtained from the FOBDD through repeated application of the following steps:
    \begin{enumerate}
        \item Merge two leafs with the same label;
        \item Merge two nonleafs $v,v'$ such that $t_V(v) = t_V(v')$, $c_{\tzero}(v) = c_{\tzero}(v')$ and $c_{\tone}(v) = c_{\tone}(v')$;
        \item Remove a nonleaf $v$ with $c_{\tzero}(v) = c_{\tone}(v)$ by rerouting all incoming edges to its unique child.
    \end{enumerate}
    \item Conversely, any OBDD obtained using these steps represents $\varphi$.
    \item When these steps can no longer be applied, the result is the ROBDD representing $\varphi$.
\end{enumerate}
\end{proposition}

As with ROBDDs, we can speak of FOBDDs representing scenarios. The following result shows that the reduced scenarios of Section \ref{ssec:subdiv} can be found in the FOBDD. It is a straightforward consequence of the definition of the Boolean function associated to an FOBDD.

\begin{lemma} \label{lem:fobdd}
Let $\Lambda$ be a scenario, and let $B$ be an FOBDD representing $(\Lambda,<)$, corresponding to a strict linear order $<$ on $F \cup A$. Suppose that $\R{B}$ is not a leaf. Let $B^{\tzero}$, $B^{\tone}$ be the sub-FOBDDs of $B$ with root $c_{\tzero}(\R{B})$ and $c_{\tone}(\R{B})$, respectively. Then $B^{\tzero}$ represents $\Lambda^{t_V(\R{B})/\tzero}$ and $B^{\tone}$ represents $\Lambda^{t_V(\R{B})/\tone}$, respectively.
\end{lemma}

This result, together with the results of Sections \ref{ssec:commute} and \ref{ssec:subdiv}, allows us to prove Theorem \ref{thm:mainmax} for FOBDDs.

\begin{theorem} \label{thm:fullmax}
Let $\hat{\Lambda}$ be a quantified scenario. Let $B = (V,E,t_V,t_E)$ be an FOBDD representing $\Lambda$. For $v \in V$, define sets $G_v,\tilde{G}_v \subseteq \mathcal{D}$ as in Theorem \ref{thm:mainmax}. Then $G_{\R{B}} = \recht{PMC}(\hat{\Lambda})$.
\end{theorem}

\begin{proof}
We prove this by induction on $|F \cup A|$. If $F \cup A = \varnothing$, then either $\varphi = \tzero$ or $\varphi = \tone$, and this is the label of $\R{B}$. Furthermore, $\Sigma_{\Lambda} = \{*\}$, and for this strategy one has $\hat{\gamma}_{\recht{m}}(*) = 0$, and $\hat{\pi}(*)$ is equal to $0$ or $1$ depending on the value of $t_V(\R{B})$. In the first case, we have
\[
G_{\R{B}} = \tilde{G}_{\R{B}} = \left\{\vvec{0\\0}\right\} = \left\{\vvec{\hat{\pi}(\hat{\gamma}_{\recht{m}}(*))\\0}\right\} = \recht{PMC}(\hat{\Lambda});
\]
the case $\hat{\pi}(*) = 1$ is analogous. This proves the base case.

Now suppose $|F \cup A| > 0$, and suppose the statement is true for every quantified scenario $\hat{\Lambda}'$ with $|F' \cup A'| < |F \cup A|$. Since $|F \cup A| > 0$, we know that $\R{B}$ is not a leaf; assume $t_V(\R{B}) = f \in F$. Then in the notation of Lemma \ref{lem:fobdd}, we now that $B^{\tzero}$ represents $\Lambda^{f/\tzero}$. Thus by the induction hypothesis we know
\[
G_{c_{\tzero}(\R{B})} = \recht{PMC}(\hat{\Lambda}^{f/\tzero}) = \PF(\Sigma_{\Lambda^{f/\tzero}}),
\]
and likewise
\[
    G_{c_{\tone}(\R{B})} = \recht{PMC}(\hat{\Lambda}^{f/\tone}) = \PF(\Sigma_{\Lambda^{f/\tone}}).
\]
Then using Propositions \ref{prop:commutef} and \ref{prop:tauf} we get
\begin{align*}
G_{\R{B}} &= \PF(\Psi_{\recht{F}}(G_{c_{\tzero}(\R{B})},G_{c_{\tone}(\R{B})},\pi_f)) \\
&= \PF(\Psi_{\recht{F}}(\PF(\Sigma_{\Lambda^{f/\tzero}}),\PF(\Sigma_{\Lambda^{f/\tzero}}),\pi_f)) \\
&= \PF(\Psi_{\recht{F}}(\Sigma_{\Lambda^{f/\tzero}},\Sigma_{\Lambda^{f/\tzero}},\pi_f)) \\
&= \PF(\Sigma_{\Lambda}).
\end{align*}
The case that $t_V(\R{B}) \in A$ is similar, except we now use Propositions \ref{prop:commutea} and \ref{prop:taua}.
\end{proof}

\subsection{Reduction to ROBDDs} \label{ssec:reduce}

In the previous section we have proven Theorem \ref{thm:mainmax} for FOBDDs. In this section, we show that this result does not change when applying the reduction steps of Proposition \ref{prop:bdd}, allowing us to conclude Theorem \ref{thm:mainmax}. The meat of this section is in the following result.

\begin{lemma} \label{lem:reduce}
Let $B$ be an ordered BDD representing a quantified scenario $\hat{\Lambda}$. Let $B'$ be obtained from $B$ via one of the operations from Proposition \ref{prop:bdd}.
Then $G_{\R{B}} = G_{\R{B'}}$.
\end{lemma}

\begin{proof}
In the second case, the computation of $G_v$ only depends on $G_{c_{\tzero}(v)}, G_{c_{\tone}(v)}$, and $t_V(v)$. Since these are all shared with $v'$, we get $G_v = G_{v'}$. Therefore, merging them does not affect the bottom-up computation. The first case is analogous.

Now consider the third case. Let $v' = c_{\tzero}(v) = c_{\tone}(v)$; we will show that $G_{v} = G_{v'}$. If this is true, then the bottom-up computation will result in the same result after removing $v$. First suppose that $t_V(v) = f \in F$. Then 
\[
G_v  = \PF\left\{\vvec{\pi_f p_{\tzero}+(1-\pi_f )p_{\tone}\\\max(c_{\tzero},c_{\tone})} \ \middle| \ \vvec{p_{\tzero}\\c_{\tzero}},\vvec{p_{\tone}\\c_{\tone}} \in G_{v'}\right\}.
\]
Let $\vvec{p_{\tzero}\\c_{\tzero}},\vvec{p_{\tone}\\c_{\tone}} \in G_{v'}$ with $p_{\tzero} \leq p_{\tone}$. Since both these vectors are Pareto optimal, they are not comparable under $\sqsubset$, so $c_{\tzero} \leq c_{\tone}$. Hence
\[
\vvec{\pi_f p_{\tzero}+(1-\pi_f )p_{\tone}\\\max(c_{\tzero},c_{\tone})} = \vvec{\pi_f p_{\tzero}+(1-\pi_f )p_{\tone}\\c_{\tone}} \sqsupseteq \vvec{p_{\tone}\\c_{\tone}}.
\]
This shows that for the Pareto front, we only need to consider the case $\vvec{p_{\tzero}\\c_{\tzero}} = \vvec{p_{\tone}\\c_{\tone}}$; but this yields precisely the elements of $G_{v'}$. We conclude that $G_v = G_{v'}$.

Now suppose that $t_V(v) = a \in A$, and let $d = \binom{0}{\gamma_a} \in \mathcal{D}$. Then
\[
G_v = \PF\left(G_{v'} \cup (G_{v'}+d)\right).
\]
If $d' \in G_{v'}$, then $d'+d \sqsupseteq d$, so the terms of $G_{v'}+d$ are not Pareto optimal. This shows that $G_v = G_{v'}$, which concludes our proof.
\end{proof}

\begin{proof}[Proof of Theorem \ref{thm:mainmax}]
By Proposition \label{prop:bdd}, we obtain the ROBDD of $\hat{\Lambda}$ from its FOBDD via reduction steps. By Theorem \ref{thm:fullmax}, the bottom-up algorithm correctly computes $\recht{PMC}(\hat{\Lambda})$, and by Lemma \ref{lem:reduce}, these reduction steps do not change the result of the computation.
\end{proof}

\end{document}